\documentclass[letterpaper, 10 pt, conference]{ieeeconf} 
\IEEEoverridecommandlockouts                              
\usepackage{graphicx}
\usepackage[letterpaper, left=0.68in, right=0.68in, top=0.75in, bottom=0.75in]{geometry}
\usepackage{svg}
\usepackage{listings}
\usepackage{subcaption}

\usepackage{etoolbox}

\usepackage{scalerel}

\usepackage[htt]{hyphenat}

\usepackage{siunitx}

\usepackage{pifont}

\usepackage{amsmath,textcomp}

\usepackage{tikz}
\usetikzlibrary{shapes, arrows, positioning, fit}
\usepackage{comment}
\allowdisplaybreaks
\usepackage{adjustbox}

\usepackage{url}

\usepackage{algpseudocode}
\usepackage{cite}
\usepackage{paralist}
\usepackage{amsthm}
\usepackage{color}
\usepackage{stmaryrd}
\usepackage{epsfig}
\usepackage{bm}
\usepackage{listings,lstautogobble}
\usepackage{mathtools}
\usepackage{amsfonts}
\usepackage{newpxtext,newpxmath}
\usepackage{xspace}
\usepackage{verbatim}
\usepackage{epstopdf}
\usepackage{parcolumns}

\usepackage[utf8]{inputenc}
\usepackage{calrsfs}
\usepackage{rotating}
\usepackage{booktabs}
\usepackage{multirow}
\usepackage{float}
\usepackage{makecell}
\usepackage{float}
\newcommand{\mycomment
}[1]{}

\usepackage[compact]{titlesec}
\titlespacing{\section}{0pt}{0.25ex}{0.25ex}
\titlespacing{\subsection}{0pt}{0.2ex}{0.2ex}
\titlespacing{\subsubsection}{0pt}{0.1ex}{0.1ex}

\usepackage[T1]{fontenc}

\DeclareMathAlphabet{\mathcal}{OMS}{cmsy}{m}{n}

\newcommand{\N}{{\mathbb{N}}}

\newcommand{\zono}[1]{\langle #1 \rangle}

\newtheorem{assumption}{Assumption}

\newtheorem{definition}{Definition}
 \newtheorem{theorem}{Theorem}
  \newtheorem{remark}{Remark}
  \newtheorem{lemma}{Lemma} 
\usepackage{algorithm}
\usepackage{amsmath}
\usepackage{algpseudocode}

  \renewcommand{\footnotesize}{\fontsize{8pt}{11pt}\selectfont}

\DeclareMathAlphabet{\mathcal}{OMS}{cmsy}{m}{n}

\usepackage{hyperref}
\hypersetup{
    colorlinks=true,
    linkcolor=black,
    urlcolor=gray,
}
\usepackage{balance}

\def\t{\mathtt{T}}

\makeatletter
\newcommand{\vast}{\bBigg@{4}}
\newcommand{\Vast}{\bBigg@{5}}
\makeatother

\makeatletter
\DeclareRobustCommand{\nand}{\mathbin{\mathpalette\n@and@or\land}}
\DeclareRobustCommand{\nor}{\mathbin{\mathpalette\n@and@or\lor}}

\DeclareRobustCommand{\enand}{\overline{\mathbin{\mathpalette\n@and@or\land}}}
\DeclareRobustCommand{\enor}{\overline{\mathbin{\mathpalette\n@and@or\lor}}}

\newcommand{\n@and@or}[2]{
  \vphantom{#2}
  \ooalign{$\m@th#1#2$\cr\hidewidth$\m@th#1\sim$\hidewidth\cr}
}
\makeatother

\begin{document}
\title{\LARGE \bf
Data-Driven Reachability of Nonlinear Lipschitz Systems via Koopman Operator Embeddings}

\author{Alireza Naderi, Ahmad Hafez, Abdulla Fawzy, and Amr Alanwar%
\thanks{All authors are with the TUM School of Computation, Information, and Technology (CIT), 
        Technical University of Munich (TUM), Germany. 
        {\tt\small (Email: \{alireza.naderi, a.hafez, abdulla.mohamed, alanwar\}@tum.de)}}%
}
        
%

\maketitle
\begin{abstract}
Data-driven safety verification of robotic systems often relies on zonotopic reachability analysis due to its scalability and computational efficiency. However, for nonlinear systems, these methods can become overly conservative, especially over long prediction horizons and under measurement noise. We propose a data-driven reachability framework based on the Koopman operator and zonotopic set representations that lifts the nonlinear system into a finite-dimensional, linear, state-input-dependent model. Reachable sets are then computed in the lifted space and projected back to the original state space to obtain guaranteed over-approximations of the true dynamics. The proposed method reduces conservatism while preserving formal safety guarantees, and we prove that the resulting reachable sets over-approximate the true reachable sets. Numerical simulations and real-world experiments on an autonomous vehicle show that the proposed approach yields substantially tighter reachable set over-approximations than both model-based and linear data-driven methods, particularly over long horizons.

\end{abstract}


\section{INTRODUCTION}
Satisfying safety constraints in dynamic and uncertain environments is a fundamental requirement in robotic control~\cite{ZHANG2025}. Traditional formal model-based verification approaches rely on accurate system dynamics and conservative uncertainty characterizations, which are often difficult to obtain for complex robotic platforms operating in real-world environments~\cite{Althoff2010,AlthoffDolan2014}. As a result, these methods may become overly conservative when precise models are unavailable. Data-driven safety verification has therefore emerged as a promising alternative, leveraging measured data to reason about system behavior directly; however, existing approaches often struggle to provide scalable reachability analysis with reliable safety guarantees under noise and modeling uncertainty.

\begin{figure}[h]
\centering
    \includegraphics[scale=0.46]{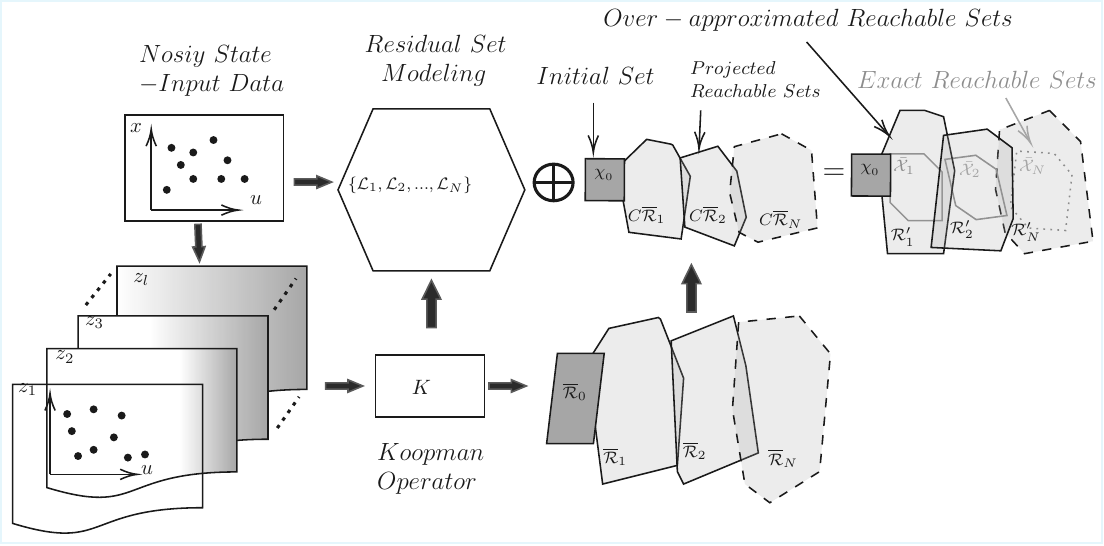}
    \caption{Over-approximated reachable sets derived from the reachable sets of the Koopman-lifted system, with residual set modeling to account for approximation errors.}
    \label{fig:koopmandiagram}
    \vspace{-7mm}
\end{figure}

Data-driven reachable set computation for safety verification in robotic and general control systems with unknown dynamics constructs over-approximations of all trajectories consistent with measured input–state data using zonotopic representations as a computationally efficient framework~\cite{10068731,alanwar2021data,akhormeh2025online,hafez2025safe,9833266}. These approaches characterize sets of models consistent with noisy measurements and propagate them forward in time, yielding reachable sets that explicitly account for bounded noise and modeling uncertainty. Zonotopic set representations rely on efficiently computable Minkowski sums and linear transformations, enabling scalable reachability analysis while maintaining conservative over-approximations of the true reachable set~\cite{girard2005reachability}. However, although these methods provide rigorous guarantees by enclosing all data-consistent trajectories, the resulting reachable sets may become overly conservative, particularly over long prediction horizons or for highly nonlinear dynamics.\par

Koopman operator theory provides a linear representation of nonlinear dynamical systems by lifting the state into a higher-dimensional space of observables~\cite{mamakoukas2019local,Str_sser_2026,proctor2018generalizing}. In this lifted space, nonlinear dynamics evolve approximately linearly, enabling improved prediction over longer horizons and facilitating analysis and control design. For instance, Koopman-based methods have demonstrated generalizable data-driven linear embeddings of nonlinear robotic systems with provable model error bounds and benefits for controller synthesis~\cite{mamakoukas2019local}.

Koopman-based linearization has also been applied to model reduction and reachability analysis of nonlinear ordinary differential equations through finite-dimensional approximations of lifted dynamics~\cite{Boreale2025}. However, such approaches typically rely on explicit knowledge of the system dynamics to derive analytical approximation-error bounds, which limits their applicability in data-driven settings where only measured trajectories are available. In contrast, data-driven Koopman reachability methods commonly employ Extended Dynamic Mode Decomposition (EDMD) or neural-network-based lifting functions to learn linear representations directly from trajectory data~\cite{nath2026scalable}. Within this framework, reachable sets are propagated in a lifted linear space and subsequently projected back onto the original state space, enabling scalable analysis of nonlinear systems. To account for model mismatch arising from learning errors, probabilistic techniques such as conformal prediction have been introduced to provide coverage guarantees at user-specified confidence levels. However, these approaches do not constitute formal methods, as they rely on statistical guarantees rather than deterministic ones.\par

Recent works have explored deep-learning-based Koopman representations for safety analysis. In~\cite{KETEMA2025100339}, a deep Koopman operator learned using Long Short-Term Memory (LSTM) networks enables safety-critical control by lifting unknown nonlinear trajectories into a linear latent space, where reachable sets are computed as zonotopic over-approximations that propagate uncertainty and noise through the learned dynamics. Despite its effectiveness, this approach suffers from low sample efficiency and high training complexity, and the learned lifting functions behave as black boxes, limiting interpretability and robustness in previously unexplored regions of the state space. While \cite{Bak2025} enables efficient reachability analysis via Koopman linearization, its guarantees apply only to the learned surrogate model and do not ensure containment of the true nonlinear trajectories. Moreover, uncertainty from data-driven identification is not explicitly accounted for, and the formulation is limited to autonomous (unforced) systems.

\emph{Contributions}: 
This paper presents a novel data-driven framework for reachability analysis of nonlinear, non-affine systems using Koopman operator embeddings, capable of handling measured state–input data corrupted by process noise. A Koopman operator is identified directly from data via least-squares regression with state-dependent and state–input-dependent lifting functions, enabling a linear representation of the underlying nonlinear dynamics. This lifted representation provides a more accurate approximation of the nonlinear system compared to conventional Linear Time-Invariant (LTI) lifted models.

To rigorously account for modeling errors, we construct data-driven over-approximations of the residual dynamics by aggregating multi-step prediction errors across all available trajectories, and extend these bounds to the entire reachable domain using a covering-radius argument. Scalable reachable sets of the lifted system are then computed by propagating convex zonotopes through the identified linear dynamics, yielding computationally efficient over-approximations that capture both process noise and linearization error.

Finally, the reachable sets are projected back to the original state space and combined with the derived modeling uncertainty sets to obtain guaranteed over-approximations of the original nonlinear system. The proposed framework bridges data-driven system identification and formal verification, providing provable safety guarantees directly from measured data. An overview of the method is illustrated in Figure~\ref{fig:koopmandiagram}. The contributions of this paper are summarized as follows:
\begin{itemize}
    \item A data-driven framework for reachability analysis of Lipschitz nonlinear, non-affine systems is introduced, leveraging Koopman operator embeddings learned directly from noisy state–input measurements.

    \item A lifting-based identification approach with state-dependent and state–input-dependent observables is developed, yielding a linear representation that improves the approximation accuracy of nonlinear dynamics compared to standard LTI models.

    \item Scalable reachable sets are computed through zonotope propagation in the lifted space and projected back to the original state space, resulting in guaranteed over-approximations of the nonlinear system behavior.

    \item We validate the proposed framework on representative nonlinear benchmark systems and a real-world autonomous vehicle, showing improved accuracy and reduced conservatism compared with existing data-driven and model-based reachability methods.
    
\end{itemize}

Readers can reproduce our results by utilizing our openly accessible repository\footnotemark, and a visual demonstration of the experiments is also available\footnotemark.

\addtocounter{footnote}{-1} 
\footnotetext{\url{https://github.com/TUM-CPS-HN/koopreach}}

\stepcounter{footnote} 
\footnotetext{\url{https://youtu.be/y6rD7-fikPc}}

The remainder of this paper is organized as follows. Section~\ref{sec:pb} presents the problem formulation and preliminaries on set representations using zonotopes. Section~\ref{sec:koopman} introduces the data-driven Koopman operator framework. Section~\ref{sec:reachability} details the proposed data-driven reachability method using a lifted system. Section~\ref{sec:evaluation} presents simulation and experimental results, and Section~\ref{sec:conclusion} concludes the paper with potential directions for future research.

\section{Problem Statement and Preliminaries} \label{sec:pb}

We start by defining some set representations that are used in the reachability analysis. 

\subsection{Set Representations}
We define the following sets:
\begin{definition}[Zonotope \cite{conf:zono1998}] \label{def:zonotopes} 
Given a center $c_{\mathcal{Z}} \in \mathbb{R}^{n_{\mathcal{Z}}}$ and $\gamma_{\mathcal{Z}} \in \mathbb{N}$ generator vectors in a generator matrix $G_{\mathcal{Z}}=\begin{bmatrix} g_{\mathcal{Z}}^{(1)}& \dots &g_{\mathcal{Z}}^{(\gamma_{\mathcal{Z}})}\end{bmatrix} \in \mathbb{R}^{n_{\mathcal{Z}} \times \gamma_{\mathcal{Z}}}$, a zonotope is defined as
\begin{equation}
	\mathcal{Z} = \Big\{ x \in \mathbb{R}^{n_{\mathcal{Z}}} \; \Big| \; x = c_{\mathcal{Z}} + \sum_{i=1}^{\gamma_{\mathcal{Z}}} \beta^{(i)} \, g^{(i)}_{\mathcal{Z}} \, ,
	-1 \leq \beta^{(i)} \leq 1 \Big\} \; .
\end{equation}
We use the shorthand notation $\mathcal{Z} = \zono{c_{\mathcal{Z}},G_{\mathcal{Z}}}$ for a zonotope. 
\end{definition}
Let $\Gamma \in \mathbb{R}^{m \times n_{\mathcal{Z}}}$ be a linear map. Then $\Gamma\mathcal{Z}= \zono{\Gamma c_{\mathcal{Z}},\Gamma G_{\mathcal{Z}}}$ \cite[p.18]{Althoff2010}.

Given two zonotopes $\mathcal{Z}_1=\langle c_{\mathcal{Z}_1},G_{\mathcal{Z}_1} \rangle$ and $\mathcal{Z}_2=\langle c_{\mathcal{Z}_2},G_{\mathcal{Z}_2} \rangle$, the Minkowski sum $\mathcal{Z}_1 \oplus \mathcal{Z}_2 = \{z_1 + z_2| z_1\in \mathcal{Z}_1, z_2 \in \mathcal{Z}_2 \}$ can be computed exactly as follows \cite{conf:zono1998}: 
\begin{equation}
     \mathcal{Z}_1 \oplus \mathcal{Z}_2 = \Big\langle c_{\mathcal{Z}_1} + c_{\mathcal{Z}_2}, [G_{\mathcal{Z}_1}, G_{\mathcal{Z}_2} ]\Big\rangle.
     \label{eq:minkowski}
\end{equation}
 We define and compute the Cartesian product of two zonotopes $\mathcal{Z}_1 $ and $\mathcal{Z}_2$ by 
\begin{align}\label{eq:cart}
\mathcal{Z}_1 \times \mathcal{Z}_2 &= \bigg\{ \begin{bmatrix}z_1 \\ z_2\end{bmatrix} \bigg| z_1 \in \mathcal{Z}_1, z_2 \in \mathcal{Z}_2 \bigg\} \nonumber\\
&= \Bigg\langle \begin{bmatrix} c_{\mathcal{Z}_1} \\ c_{\mathcal{Z}_2} \end{bmatrix}, \begin{bmatrix} G_{\mathcal{Z}_1} & 0 \\ 0 & G_{\mathcal{Z}_2}\end{bmatrix} \Bigg\rangle.
\end{align}

\subsection{Problem Statement}
We consider a discrete-time system
\begin{align}
\begin{split}
    x_{k+1} &= f(x_k,u_k)+ w_k,\\
\end{split}
    \label{eq:sysnonlingen}
\end{align}
where $f:\mathbb{R}^{n_x}\times\mathbb{R}^{n_u} \rightarrow \mathbb{R}^{n_x}$ a twice differentiable unknown function, $w_k \in \mathcal{Z}_w \subset \mathbb{R}^{n_x}$ denotes the noise bounded by a noise zonotope $\mathcal{Z}_w$, ${u_k \in \mathcal{U}_k \subset \mathbb{R}^{n_u}}$ the input bounded by an input zonotope $\mathcal{U}_k$, and ${x_0 \in \mathcal{X}_0 \subset \mathbb{R}^{n_x}}$ the initial state of the system bounded by the initial set $\mathcal{X}_0$. 

\begin{assumption}[Lipschitz Condition]
It holds that $f: \mathcal{F} \rightarrow \mathbb{R}^{n_x}$ is Lipschitz continuous, i.e., that there is some $L^\star \geq 0$ such that 
$\| f(z_k) - f(z_k^{\prime}) \|_2 \leq L^\star \| z_k - z_k^{\prime}\|_2$
holds for all $z_k, z_k^{\prime} \in \mathcal{F}$.
\label{as:lipschitz}
\end{assumption}

Reachability analysis computes the set of states $x_k$ which can be reached given a set of uncertain initial states $\mathcal{X}_0$ and a set of uncertain inputs $\mathcal{U}_k$. More formally, it can be defined as follows:
\begin{definition} [Exact Reachable Set]
The exact reachable set $\mathcal{X}_{N}$ after $N$ time steps subject to inputs ${u_k \in \mathcal{U}_k}$, $\forall k {=}\{ 0, \dots, N-1\}$, and noise $w_k \in \mathcal{Z}_w$, is the set of all states trajectories starting from initial set $\mathcal{X}_0$ after $N$ steps 
\begin{align} \label{eq:R}
        \bar{\mathcal{X}}_{N} = \big\{& x_N \in \mathbb{R}^{n_x} \, \big| x_{k+1} = f(x_k,u_k) + w_k, \nonumber\\
       & \, x_0 \in \mathcal{X}_0,
        u_k \in \mathcal{U}_k, w_k \in \mathcal{Z}_w: \nonumber\\ & \forall k \in \{0,...,N{-}1\}\big\}.
\end{align}
\end{definition}
We aim to compute an over-approximation of the exact reachable sets when the model of the system in \eqref{eq:sysnonlingen} is unknown, but input and noisy state trajectories are available.\par
Instead of having access to a mathematical model of the system, we consider $q$ input-state trajectories of different lengths $T_i+1$, denoted by $\{u^{(i)}_k\}_{k=0}^{T_i - 1}$, and $\{x^{(i)}_k\}_{k=0}^{T_i}$, $i=1, \dots, q$. We collect the set of all data sequences in the following matrices
\begin{align*}
     X &=  \begin{bmatrix} x^{(1)}_0 \dots  x^{(1)}_{T_1}  \dots  x^{(q)}_{0} \dots  x^{(q)}_{T_q}\end{bmatrix}.
 \end{align*}
Let us further denote the shifted signals
 \begin{align*}
     X_+ &= \begin{bmatrix} x^{(1)}_1\dots  x^{(1)}_{T_1} \dots  x^{(q)}_{1}  \dots  x^{(q)}_{T_q} \end{bmatrix}, \nonumber\\
     X_- &= \begin{bmatrix} x^{(1)}_0 \dots  x^{(1)}_{T_1\!-\!1}  \dots  x^{(q)}_0  \dots  x^{(q)}_{T_q\!-\!1} \end{bmatrix},\\
    U_- &= \begin{bmatrix} u^{(1)}_0  \dots  u^{(1)}_{T_1\!-\!1} \dots u^{(q)}_0  \dots  u^{(q)}_{T_q\!-\!1} \end{bmatrix}.
 \end{align*}
The total amount of data points from all available shifted signals is denoted by $T = \sum_{i=1}^{q} T_i$, and we denote the set of all available data by $D=(U_-,X)$.

\subsection{Noise Zonotope and Notations}

We define the unknown bounded noise set with Assumption \ref{ass:zon-noise}.
\begin{assumption}[Noise Zonotope]\label{ass:zon-noise}
The noise $w_k$ is contained in a known zonotope for all $k \in \mathbb{Z}_{\geq 0}$, 
i.e., $w_k \in \mathcal{Z}_w = \langle c_{\mathcal{Z}_w}, G_{\mathcal{Z}_w} \rangle$, 
where $c_{\mathcal{Z}_w} \in \mathbb{R}^{n_{\mathcal{Z}_w}}$ is the center, 
$G_{\mathcal{Z}_w} \in \mathbb{R}^{n_{\mathcal{Z}_w} \times \gamma_{\mathcal{Z}_w}}$ is the generator matrix. 
\end{assumption}
The set of real and natural numbers are denoted as $\mathbb{R}$ and $\mathbb{N}$, respectively, and $\mathbb{N}_0 = \mathbb{N}\cup \{0\}$. The transpose and Moore-Penrose pseudoinverse of a matrix $X$ are denoted as $X^\t$ and $X^\dagger$, respectively. We denote the Kronecker product by $\otimes$. We denote the element at row $i$ and column $j$ of matrix $A$ by $(A)_{i,j}$ and column $j$ of $A$ by $(A)_{.,j}$.  
For a list or vector of elements, we denote the element $i$ of vector or list $a$ by $a^{(i)}$. 
For a matrix $A = [x^{1}_0,..., x^{1}_{m_1}, \dots,x^{n_A}_0,...,x^{n_A}_{m_A}]$, we denote by $A^{(i)} = [x^{i}_0, x^{i}_1, \dots, x^{i}_{m_A}]$ the $i$-th column vector of $A$. The element-wise multiplication of two matrices is denoted by $\odot$. For a matrix $A \in \mathbb{R}^{n_A \times m_A}$, the Frobenius norm is defined as $\|A\|_F = \sqrt{\sum_{i=1}^{n_A} \sum_{j=1}^{m_A} |a_{ij}|^2}$. For a vector $x \in \mathbb{R}^{n}$, the Euclidean norm is denoted by $\|x\| = \sqrt{x^\top x}$. 
We denote the over-approximation of a reachable set $\bar{\mathcal{X}}_k$ by an interval by $\text{int}(\bar{\mathcal{X}}_k)$. 
We define also for $N$ time steps
\begin{align}
    \mathcal{F} = \cup_{k=0}^{N} (\bar{\mathcal{X}}_k \times \mathcal{U}_k).
    \label{eq:F}
\end{align}

\section{Koopman Operator Representation}\label{sec:koopman}
In this section, the Koopman operator framework is introduced for representing nonlinear dynamical systems through linear embeddings in a higher-dimensional space. To illustrate the Koopman operator main idea, consider the following simple nonlinear system
\begin{align}   
\begin{bmatrix}\label{eq:nonexample}
{x}_{1,k+1} \\
{x}_{2,k+1}
\end{bmatrix}
=
\begin{bmatrix}
\mu x_{1,k} \\
\lambda (x_{2,k} - x_{1,k}^2) + \delta u_k
\end{bmatrix},
\end{align}
where $\lambda, \mu, \delta \in \mathbb{R}$. Define the lifting (observable) functions as $y_{1,k} = x_{1,k}$, $y_{2,k} = x_{2,k}$, and $y_{3,k} = x_{1,k}^2$. In terms of these lifted coordinates, system \eqref{eq:nonexample} can be equivalently expressed as the following linear system
\begin{align} \label{eq:koopmanexample}
\begin{bmatrix}
{y}_{1,k+1} \\
{y}_{2,k+1} \\
{y}_{3,k+1}
\end{bmatrix}
=
\begin{bmatrix}
\mu & 0 & 0 \\
0 & \lambda & -\lambda \\
0 & 0 & \mu^2
\end{bmatrix}
\begin{bmatrix}
y_{1,k} \\
y_{2,k} \\
y_{3,k}
\end{bmatrix}
+
\begin{bmatrix}
0 \\
\delta \\
0
\end{bmatrix}
u_k.
\end{align}

System \eqref{eq:koopmanexample} represents a linear lifted system in the space of observables, where the nonlinear dynamics of \eqref{eq:nonexample} are captured through the lifting functions $y_{1,k}$, $y_{2,k}$, and $y_{3,k}$. This illustrates how the Koopman operator framework enables a linear representation of nonlinear dynamics in an augmented state space.

Consider the discrete-time nonlinear system in~\eqref{eq:sysnonlingen}, where \(x_k \in \mathcal{X}_k\) evolves on a manifold \(\mathcal{X}_k \subset \mathbb{R}^{n_x}\). Let \(\psi : \mathcal{X}_k \times \mathcal{U}_k \to \mathbb{R}\) denote a \emph{lifting function}. These lifting functions form an infinite-dimensional Hilbert space \(\mathcal{H}\).

The \emph{Koopman operator} \(\mathcal{K} : \mathcal{H} \to \mathcal{H}\) is a linear operator that advances all scalar-valued lifting functions \(\psi \in \mathcal{H}\) by one time step \cite{dahdah2022system}. Applying the Koopman operator to the nonlinear system \eqref{eq:sysnonlingen} yields
\begin{equation}\label{eq:koopmanlift_noise}
(\mathcal{K} \psi)(x_k, u_k) = \psi(f(x_k, u_k) + w_k, \tilde{u}_k),
\end{equation}
where \(\tilde{u}_k = u_k\) if the input exhibits state-dependent dynamics, or \(\tilde{u}_k = 0\) if the input is exogenous \cite{dahdah2022system}. Inputs generated by a controller are typically considered state-dependent.

Let the vector-valued lifting function \(\psi : \mathcal{X}_k \times \mathcal{U}_k \to \mathbb{R}^p\) be partitioned as
\begin{equation}\label{eq:lift}
\psi(x_k, u_k) =
\begin{bmatrix}
\phi(x_k) \\
\nu(x_k, u_k)
\end{bmatrix},
\end{equation}
where \(\phi : \mathcal{X}_k \to \mathbb{R}^{p_\phi}\), \(\nu : \mathcal{X}_k \times \mathcal{U}_k \to \mathbb{R}^{p_\nu}\), and \(p_\phi + p_\nu = p\). When the input is exogenous (\(\tilde{u}_k = 0\)), \eqref{eq:koopmanlift_noise} can be expressed as \cite{dahdah2022system}
\begin{equation}\label{eq:koopmn_psi}
\phi(x_{k+1}) = K \psi(x_k, u_k) + r_{\psi_k} + r_{w_k},
\end{equation}
where
\begin{equation*}
r_{w_k} = \psi(f(x_k, u_k)+w_k)-\psi(f(x_k, u_k)),
\end{equation*}
and \(r_{\psi_k}\) is the residual due to Koopman linearization and $K$ is approximation of Koopman operator. Defining \([A \; B] = K\), \eqref{eq:koopmn_psi} can be written as
\begin{equation}\label{eq:nonaffinekoopman}
\phi(x_{k+1}) = A \phi(x_k) + B \nu(x_k, u_k) + r_{\psi_k} + r_{w_k}.
\end{equation}

\begin{remark}[Bilinear Representation]
In many robotic applications, systems are governed by nonlinear control-affine dynamics \cite{iacob2024koopman}. A commonly used simplified version of \eqref{eq:sysnonlingen} assumes an affine structure, given by
\begin{equation}\label{eq:sysaffine}
x_{k+1} = f(x_k) + g(x_k) u_k + w_k.
\end{equation}
A common choice for the lifting function is
\begin{equation}\label{eq:bilinearlift}
\nu(x_k, u_k) = u_k \otimes \phi(x_k).
\end{equation}
Substituting \eqref{eq:bilinearlift} into \eqref{eq:lift} and applying the Koopman operator to \eqref{eq:sysaffine} yields
\begin{equation}
\phi(x_{k+1}) = \bar{A} \phi(x_k) + \sum_{i=1}^{n_u} \bar{B}_i \phi(x_k) u_k^{(i)} + \bar{r}_k + \bar{r}_{w_k},
\end{equation}
where \(\bar{B}_i \in \mathbb{R}^{p_\phi \times p_\phi}\),  \(\bar{r}_k\) is the residual and
\begin{equation*}
\bar{r}_{w_k} = \psi(f(x_k)+g(x_k) u_k+w_k)-\psi(f(x_k)+g(x_k) u_k).
\end{equation*}
\end{remark}

\begin{remark}[LTI Representation]\label{re:ltiKoopman}
Another common simplification of \eqref{eq:sysnonlingen} is
\begin{equation}\label{eq:sysaffine_Bcont}
x_{k+1} = f(x_k) + b u_k + w_k,
\end{equation}
where \(b \in \mathbb{R}^{n_x \times n_u}\) is constant. Selecting \(\nu(x_k, u_k) = u_k\) and applying the Koopman operator yields
\begin{equation}\label{eq:koopmanLTI}
\phi(x_{k+1}) = \tilde{A} \phi(x_k) + \tilde{B} u_k + \tilde{r}_k + \tilde{r}_{w_k},
\end{equation}
where \(\tilde{B} \in \mathbb{R}^{p_\phi \times n_u}\) \(\tilde{r}_k\) is the residual and
\begin{equation*}
\tilde{r}_{w_k} = \psi(f(x_k)+b u_k+w_k) - \psi(f(x_k)+b u_k).
\end{equation*}
\end{remark}

\subsection{Approximating the Koopman Operator From Data}
Since the Koopman operator is infinite-dimensional, a finite-dimensional approximation must be considered for numerical implementation. Such an approximation forms the foundation of most data-driven approaches used to compute the operator’s spectral properties.\par

To approximate the Koopman matrix from a dataset $D=(U_-,X)$, consider a lifted data snapshot as follows:
\begin{align*}
\Phi(X_+) &= 
\begin{bmatrix}
\phi(x^{(1)}_1), \dots, \phi(x^{(1)}_{T_1}), \dots, 
\phi(x^{(q)}_1), \dots, \phi(x^{(q)}_{T_q})
\end{bmatrix}, \\
\Phi(X_-) &= 
\begin{bmatrix}
\phi(x^{(1)}_0), \dots, \phi(x^{(1)}_{T_1-1}), \dots, 
\phi(x^{(q)}_0), \dots, \phi(x^{(q)}_{T_q-1})
\end{bmatrix}.
\end{align*}
The combined lifted data is then
\begin{align*}   
 \Psi(X_-,U_-) = 
\begin{bmatrix}\Phi(X_-)\\N(X_-,U_-)\end{bmatrix},
\end{align*}
where
\begin{align*}
N(X_-,U_-) &= 
\Big[ 
\nu(x^{(1)}_0,u^{(1)}_0), \dots, \nu(x^{(1)}_{T_1-1},u^{(1)}_{T_1-1}), \dots\\ 
& \quad \nu(x^{(q)}_0,u^{(q)}_0), \dots, \nu(x^{(q)}_{T_q-1},u^{(q)}_{T_q-1})
\Big].
\end{align*}
Then the Koopman matrix is computed through the following minimization problem

\begin{align}\label{eq:mse_koopman}
K^\star
= \arg\min_{K}
\bigl\| \Phi(X_+) - K\Psi(X_-,U_-) \bigr\|_F^2 .
\end{align}
The Least-Squares (LS) approach provides the most straightforward solution for approximating the Koopman operator to solve the optimization problem in \eqref{eq:mse_koopman} as follows 
\begin{align}\label{eq:ls_koopman}
K^\star
=  \Phi(X_+)\Psi(X_-,U_-)^\dagger  .
\end{align}
However, the LS approximation of the Koopman operator is often affected by numerical and computational limitations. In particular, computing the pseudoinverse of $\Psi(X_-,U_-)$ becomes computationally expensive when the dataset contains a large number of snapshots. To address this issue, the EDMD method reduces the dimensionality of the pseudoinverse required to compute \eqref{eq:mse_koopman} when the number of snapshots is significantly larger than the dimension of the lifted state \cite{kutz2016dynamic,Williamsetal2015}. 

\begin{remark}[Lifting Function Selection]\label{re:liftingfunction}
The choice of lifting (observable) functions is fundamental in the Koopman operator framework, as it directly affects the ability to represent nonlinear dynamics in a finite-dimensional linear form. Well-chosen observables can capture the essential structure of the underlying system and yield accurate approximations. However, in general, no systematic method exists for selecting an optimal set of lifting functions, and the choice remains problem-dependent. Common selections include polynomial bases, orthogonal functions, radial basis functions, and other nonlinear transformations of the state and input. While enriching the set of observables typically improves approximation accuracy, it also increases computational complexity, resulting in a trade-off between model fidelity and scalability.
\end{remark}

\section{Data-Driven Reachability Analysis}\label{sec:reachability}
In this section, we aim to compute an over-approximation of the exact reachable set \eqref{eq:R} of the original nonlinear system. This is achieved by conducting reachability analysis on the lifted system \eqref{eq:koopmn_psi} using the approximated Koopman operator. 
\begin{assumption}[State-Inclusive Lifting]
The lifting function $\phi:\mathcal{X}_k\to\mathbb{R}^{p_\phi}$ includes the
original system states as components, i.e., $x_k = C \phi(x_k)$, for some matrix $C \in \mathbb{R}^{n_x \times {p_\phi}}$.
\end{assumption}

\begin{assumption}[Uniformly Lipschitz and Differentiable Lifting]
\label{ass:lipschitz}
The lifting function $\psi:\mathcal{X}_k\times\mathcal{U}_k\to\mathbb{R}^p$
is continuously differentiable and Lipschitz continuous with respect to $x_k$,
uniformly over $u_k\in\mathcal{U}_k$. That is, there exists a constant
$L_\psi>0$ such that
\begin{align*}
\|\psi(x_{1,k},u_k)-\psi(x_{2,k},u_k)\|
&\le L_\psi \|x_{1,k}-x_{2,k}\|, \\
&\quad \forall x_{1,k},x_{2,k}\in\mathcal{X}_k,\ \forall u_k\in\mathcal{U}_k.
\end{align*}
\end{assumption}

\begin{lemma}[Lifted Noise Over-approximation]
\label{lm:liftednoiseapp}
Let the lifting function $\psi:\mathcal{X}_k\times\mathcal{U}_k\to\mathbb R^p$
be Lipschitz continuous in $x_k$, uniformly in $u_k$, with constant $L_\psi$.
Consider the disturbed dynamics in \eqref{eq:sysnonlingen}, then the perturbation induced in the lifted coordinates satisfies
\begin{align*}
\|\psi(f(x_k,u_k)+w_k,u_k)
-\psi(f(x_k,u_k),u_k)\|
\le L_\psi \|w_k\|.
\end{align*}
Consequently, if $w_k\in\mathcal Z_w=\zono{c_{\mathcal Z_w},G_{\mathcal Z_w}}$, the lifted residual satisfies
\[
r_{w_k}\in
\mathcal Z_{w_\psi}
=\zono{0,\,L_\psi\rho_w I_p},
\]
where
\[
\rho_w=
\|c_{\mathcal Z_w}\|+
\sum_{i=1}^{n_{\mathcal{Z}_w}}\|g^{(i)}_{\mathcal{Z}_w}\|.
\]
\end{lemma}

\begin{proof}
Let $\bar x_{k+1}=f(\bar{x}_k,u_k)$. Then
$x_{k+1}=\bar x_{k+1}+w_k$.
Since $\psi$ is Lipschitz continuous in $x_k$, uniformly in $u_k$, we have
\[
\|\psi(\bar x_{k+1}+w_k,u_k)
-\psi(\bar x_{k+1},u_k)\|
\le L_\psi\|w_k\|.
\]
Hence, the lifted perturbation is bounded by $L_\psi\|w_k\|$.

If $w_k\in\mathcal Z_w=\zono{c_{\mathcal Z_w},G_{\mathcal Z_w}}$, then according to Definition \ref{def:zonotopes}, we have  
\[
w_k=c_{\mathcal Z_w}+\sum_{i=1}^{n_{\mathcal{Z}_w}}\beta^{(i)}g^{(i)}_{\mathcal{Z}_w},
\qquad -1 \leq \beta^{(i)} \leq 1.
\]
We obtain the following upper bound on $\|w_k\|_2$: 
\[
\|w_k\|
\le
\|c_{\mathcal Z_w}\|
+
\left\|
\sum_{i=1}^{n_{\mathcal{Z}_w}}\beta^{(i)}g^{(i)}_{\mathcal{Z}_w}
\right\|.
\]
Using the triangle inequality and $|\beta^{(i)}|\le1$ gives
\[
\|w_k\|
\le
\|c_{\mathcal Z_w}\|
+
\sum_{i=1}^{n_{\mathcal{Z}_w}}\|g^{(i)}_{\mathcal{Z}_w}\|
=
\rho_w.
\]
Therefore,
\[
\|r_{w_k}\|
\le
L_\psi\rho_w,
\]
which yields the stated zonotopic over-approximation.
\end{proof}


The following theorem establishes that the reachable sets computed by Algorithm \ref{alg:LipReachability} constitute an over-approximation of the exact reachable set.

\begin{theorem}[Reachable set over-approximation]
\label{th:reachdisnonlin}
Let $D = (U_-,X)$ be a dataset generated by the nonlinear system \eqref{eq:sysnonlingen}. 
Assume Assumption~\ref{ass:lipschitz} holds and that the sets $\bar{\mathcal{L}}_k$, $\mathcal{L}_k$, $\mathcal{Z}_{w_\psi}$, and $\mathcal{Z}_\epsilon$ respectively bound the linearization error of the nonlinear lifted system, the Koopman modeling residual computed over the sampled trajectories, the lifted process noise, and the uncertainty arising from the finite sampling density of the dataset. Then the reachable sets $\mathcal R'_k$ computed by Algorithm~\ref{alg:LipReachability} satisfy $\bar{\mathcal X_k} \subseteq \mathcal R'_k, \qquad \forall k=1,\dots,N.$
\end{theorem}
\begin{proof}
We prove the result in four steps by separately over-approximating each source of uncertainty with a zonotopic set, and then combining these bounds in the reachability recursion.

Let's start with computing the over-approximation set of the linearization error of the nonlinear lifted system $\bar{\mathcal{L}}_k$ by defining $z_k = \phi(x_k)$. Then, from \eqref{eq:nonaffinekoopman} we obtain
\begin{align} \label{eq:gkoopman}
     g(z_k,u_k)=A z_k + B \nu(z_k, u_k).
\end{align}
Since the lifting function is continuously differentiable under Assumption~\ref{ass:lipschitz}, the function $g(z_k,u_k)$ is also continuously differentiable. A local linearization of \eqref{eq:gkoopman} is obtained via a Taylor series expansion around the linearization point $\bar{z}^\star=\begin{bmatrix}z^\star \\u^\star \end{bmatrix}$,
\begin{align*}
g(\bar{z}_k) =& g(\bar{z}^\star) + 
\frac{\partial g(\bar{z}_k)}{\partial \bar{z}_k}\Big|_{\bar{z}_k=\bar{z}^\star}
(\bar{z}_k - \bar{z}^\star)+ \dots .
\end{align*}

The infinite Taylor expansion \cite{conf:taylor} can be written as a first-order approximation plus a Lagrange remainder term $L(\bar{z})$ \cite[p.65]{Althoff2010}:
\begin{align}
g(\bar{z}_k) = g(\bar{z}^\star) +
\frac{\partial g(\bar{z}_k)}{\partial \bar{z}_k}\Big|_{\bar{z}_k=\bar{z}^\star}
(\bar{z}_k - \bar{z}^\star)+\bar{L}(\bar{z}_k).
\label{eq:linfL}
\end{align}

Rewriting \eqref{eq:linfL} yields
\begin{align*}
g(z_k,u_k) =& g(z^\star,u^\star)
+ \underbrace{\frac{\partial g}{\partial z_k}\Big|_{z^\star,u^\star}}_{\bar{A}} (z_k - z^\star)
\nonumber\\
&+ \underbrace{\frac{\partial g}{\partial u_k}\Big|_{z^\star,u^\star}}_{\bar{B}} (u_k - u^\star)
+ \bar{L}(z_k,u_k),
\end{align*}
which can be expressed compactly as
\begin{align}
g(z_k,u_k) =
M \begin{bmatrix}1\\z_k-z^\star\\ u_k-u^\star\end{bmatrix}
+\bar{L}_k,
\label{eq:fz_incl}
\end{align}
where $M=\begin{bmatrix}g(z^\star,u^\star) & \bar{A} & \bar{B}\end{bmatrix}$.
The matrices $\bar{A}$ and $\bar{B}$ follow from \eqref{eq:gkoopman}:
\begin{align*}\label{eq:koopmansysmatrix}
\bar{A} =& A + B \frac{\partial \nu}{\partial z_k}\Big|_{z^\star,u^\star}, \,\,  \,\, 
\bar{B} = B \frac{\partial \nu}{\partial u_k}\Big|_{z^\star,u^\star}.
\end{align*}

The Taylor remainder in \eqref{eq:fz_incl} can be computed as $\bar{{L}}_k$ \cite{Althoff2010}:
\begin{equation}
\bar{{L}}_k =
\frac{1}{2}(\bar{z}_k-\bar{z}^\star)^\top H(\xi_k)(\bar{z}_k-\bar{z}^\star),
\end{equation}
where $\xi_k \in \Big\{ \bar{z}^\star + \alpha (\bar{z}_k - \bar{z}^\star) \;\big|\; \alpha \in [0,1] \Big\}$ and $H(\xi_k)$ denotes the Hessian of the nonlinear dynamics. Hence
$\bar{L}_k\in \bar{\mathcal{L}}_k$ only bounds the linearization error of nonlinear lifted system \eqref{eq:gkoopman} (see Remark \ref{re:Llinkoop}). In practice, following~\cite{althoff2015introduction,Althoff2010}, $\bar{\mathcal{L}}_k$ is computed by evaluating interval bounds of the Hessian over the interval hull of the reachable set $\bar{\mathcal{R}}_k$. This ensures that $\bar{L}_k \in \bar{\mathcal{L}}_k$ rigorously bounds the linearization error of the lifted system \eqref{eq:gkoopman}, even when the exact Hessian at individual points $\xi_k$ is unknown.

 To bound the Koopman modeling residual error from data, denoted by $\mathcal{L}_k$, we consider \eqref{eq:gkoopman}, which can be written as
\begin{align} \label{eq:zgkoopman}
     z_{k+1}=A z_k + B \nu(z_k, u_k).
\end{align}
We evaluate the difference between the multi-step predictions of the nonlinear lifted Koopman model in \eqref{eq:zgkoopman} and the true lifted trajectories in the dataset over a prediction horizon $N$. This procedure is carried out in Algorithm~\ref{alg:LipReachability}, lines~\ref{alg:startfor}--\ref{alg:endfor}, where residuals are computed for all admissible trajectories in the lifted dataset. Subsequently, the computed residuals are over-approximated using interval zonotopes, as implemented in lines~\ref{alg:startforinterval}--\ref{alg:endforinterval} of Algorithm~\ref{alg:LipReachability}. The resulting sets $\mathcal{L}_k$, defined in line~\ref{alg:Lk}, provide a data-driven over-approximation of the Koopman modeling error over the prediction horizon $N$.

Since a reachable state may not coincide exactly with a sampled point, an additional uncertainty due to the finite coverage of the dataset must be taken into account. The covering radius $\delta$ guarantees that every state $s_k=[x_k^\top,u_k^\top]^\top\in\mathcal{F}$ admits a neighboring sample $s_i\in D$ satisfying $\|s_k-s_i\|\le\delta$. Hence, by Assumption~\ref{as:lipschitz},
\[
\|f(s_k)-f(s_i)\|
\le
L^\star\|s_k-s_i\|
\le
L^\star\delta.
\]
This additional uncertainty is captured by the zonotope
\[
\mathcal{Z}_\epsilon =
\zono{0,\mathrm{diag}(L^\star\delta/2,\ldots,L^\star\delta/2)},
\]
which accounts solely for the finite sampling density of the dataset.
The effect of noise on the lifted system is bounded using Lemma~\ref{lm:liftednoiseapp}. In particular, the lifted disturbance satisfies
\[
r_{w_k} \in \mathcal Z_{w_\psi}.
\]
Finally, we propagate the reachable set of the linearized lifted system subject to noise, as computed in line~\ref{ln:alglipRprime} of Algorithm~\ref{alg:LipReachability}. The resulting reachable set in the lifted space is then projected onto the original state space using $C=[I_{n_x}~0]$. Subsequently, a Minkowski sum with the Koopman modeling residual error set and the dataset covering error set is performed (cf. Algorithm~\ref{alg:LipReachability}), yielding an over-approximation of the exact reachable set of the nonlinear system:
\begin{align}\label{eq:overapproximationoriginal}
x_{k+1} \in
C \Big(
M \begin{bmatrix}1\\z_k-z^\star\\ u_k-u^\star\end{bmatrix}
\oplus \bar{\mathcal{L}}_k
\oplus \mathcal Z_{w_\psi}
\Big)
\oplus \mathcal{L}_k
\oplus \mathcal{Z}_\epsilon.
\end{align}
Therefore, the reachable sets generated by Algorithm~\ref{alg:LipReachability} satisfy
\[
\bar{\mathcal X}_k \subseteq \mathcal R'_k, \qquad \forall k,
\]
which completes the proof.
\end{proof}
\begin{remark}\label{re:Llinkoop}
We linearize \eqref{eq:gkoopman} in order to enable the use of zonotopes for reachable set propagation, rather than to approximate the Koopman modeling error. Consequently, we introduce $\bar{\mathcal{L}}_k$ to over-approximate the linearization remainder. For example, in the case of a Koopman LTI representation (see Remark~\ref{re:ltiKoopman}), where $\nu(x_k,u_k)=u_k$, the linearization error satisfies $\bar{\mathcal{L}}_k=0$, since the lifted model is fully linear.
\end{remark}

\begin{remark}[Computing $L^\star$ and $\delta$]
To estimate the Lipschitz constant $L^\star$ and the covering radius $\delta$, we follow the method presented in \cite[Remark~9]{10068731} using the available data $D$.
\end{remark}

\begin{algorithm}[t]
  \caption{Reachability Analysis of the Lifted System}
  \label{alg:LipReachability}
  \textbf{Input}: input-state trajectories $D = (U_-,X)$, initial set $\mathcal{X}_{0}$, process noise zonotope $\mathcal{Z}_{w}$ and Lipschitz constants $L$ and $L_{\psi}$, covering radius $\delta$, and input zonotope $\mathcal{U}_k$, $\forall k = 0, \dots,N-1$.\\
  \textbf{Output}: reachable sets $\bar{\mathcal{R}}_{k}, \forall k = 1, \dots,N$. 
  \begin{algorithmic}[1]
    \State $\bar{\mathcal{R}}_{0} =\phi(\mathcal{X}_{0})$,
    \State $\mathcal{Z}_\epsilon = \zono{0,\textup{diag}({L}^{(1)} \delta/2,\dots,{L}^{(n_x)} \delta/2)}$, \label{ln:alglipZeps}
    
    \For{Each trajectory $i$ }\label{alg:startfor}
        \For{Each column $j$}
            \State Initialization: $\hat{z}_1 = (Z_-)_{.,j}^{(i)}$,
            \For{Each prediction step $k$ until $N$}
                \State Prediction: $\hat z_{k+1} = A \hat z_k + B \, \nu(\hat z_k, (U_{-})^{(i)}_{.,j+k})$,
                \State Residual: $R_i^{(k)} = C \big( (Z_+)_{.,j+k}^{(i)} - \hat z_{k+1} \big)$,
            \EndFor
        \EndFor
    \EndFor    \label{alg:endfor}
    \For{Each prediction step $k$ until $N$}\label{alg:startforinterval}
        \State $\underline{l}_k = \underset{i}{\min} R_i^{(k)}$,
        \State $\overline{l}_k = \underset{i}{\max} R_i^{(k)}$,
        \State $\mathcal{L}_{k} = \text{zonotope}(\underline{l}_k,\,\overline{l}_k)$,\label{alg:Lk}
    \EndFor \label{alg:endforinterval}
    
    \For{$k = 0:N-1$}
        \State Update:~$M \gets (x^\star, u^\star)$,
        \State $\bar{\mathcal{R}}_{k+1} = M \Big(( \bar{\mathcal{R}}_{k} -x^\star )\times (\mathcal{U}_k - u^\star )\Big) \oplus \bar{\mathcal{L}}_{k} \oplus  Z_{w_\psi}$, \label{ln:alglipRprime}
        \State $\mathcal{R}^\prime_{k+1}= C\bar{\mathcal{R}}_{k+1} \oplus \mathcal{L}_{k} \oplus \mathcal{Z}_\epsilon$.\label{ln:algkooporiginal}
    \EndFor
  \end{algorithmic}
\end{algorithm}

\section{Evaluation}\label{sec:evaluation}
The performance of the proposed method for over-approximating the exact reachable set is evaluated through two numerical examples and a real-world experiment.

\subsection{Example 1: Affine Lipschitz dynamic system}
We consider a scenario where we have collected data and we do not know the underlying system type. We apply the proposed data-driven reachability analysis to a discrete stirred tank reactor (CSTR) simulation model \cite{conf:nonlinearexample} with the following dynamic
\begingroup
\makeatletter
\def\f@size{9}\check@mathfonts
\def\maketag@@@#1{\hbox{\m@th\large\normalfont#1}}%
\begin{align*}
x_{1,k}=& 
\frac{\Big(1-\frac{qT}{2V} - k_0 T_s \exp\!\left(-\frac{E_R}{x_{2,k}+T_0}\right)\Big)(x_{1,k}+C_{A0})
+ \frac{q}{V} C_{Af} T_s}
{1+\frac{qT_s}{2V}}\\
&\quad + u_{1,k}T_s - C_{A0}, \\[6pt]
x_{2,k} =& 
\frac{(x_{2,k}+T_0)\Big(1-\frac{qT_s}{2V} - \frac{T_s UA}{2V\rho C_p}\Big)
+ T_s\Big(\frac{q}{V}T_f + \frac{UA}{V\rho C_p} U_{\mathrm{ctrl}}\Big)}
{1+\frac{qT_s}{2V}+\frac{T_s UA}{2V\rho C_p}} \\
&\quad
- \frac{(x_{1,k}+C_{A0})\,\Delta H\, k_0 T_s}{\rho C_p}
\exp\!\left(-\frac{E_R}{x_{2,k}+T_0}\right)
+ u_{2,k}T_s - T_0.
\end{align*}
\endgroup
The parameters are defined as follows: $\rho = 1000$, $C_p = 0.239$, $\Delta H = -5\times 10^4$, $E_R = 8750$, $k_0 = 7.2\times 10^{10}$, $UA = 5\times 10^4$, $q = 100$, $V = 100$, $T_f = 350$, $C_{Af} = 1$, $C_{A0} = 0.5$, $T_0 = 350$, $T_{c0} = 300$, and $T_s = 0.015$. The initial set is a zonotope $\mathcal{X}_0 =\zono{\begin{bmatrix} 0.5 & 0.4 \end{bmatrix}^\t,\text{diag}(\begin{bmatrix} 0.05  & 0.3 \end{bmatrix})}$. The input set $\mathcal{U}_k =\zono{\begin{bmatrix}-0.4 & 1.4\end{bmatrix}^\t,\text{diag}(\begin{bmatrix} 0.1 & 0.2 \end{bmatrix})}$ and the noise set $\mathcal{Z}_w=\zono{0,\begin{bmatrix}1 \times 10^{-4} & 1 \times 10^{-4}\end{bmatrix}^\t}$. We employ second-order polynomial lifting functions with bias terms defined as
\begin{align*}
\psi(x) &= 
[1,\, x_1,\, x_2,\, x_1^2,\, x_1 x_2,\, x_2^2]^\top, \\
\nu(x,u) &= 
[\,u_1,\, u_2]^\top.
\end{align*}
In Figure~\ref{fig:example0}, the reachable sets computed using CORA v2025 \cite{althoff2015introduction}, a least-squares (LS) data-driven model \cite{alanwar2021data}, and the proposed Koopman-based method all over-approximate the exact reachable sets (obtained via random simulations). However, the proposed method shows significantly reduced conservatism over longer prediction horizons.

\begin{figure}[thp]
    \centering
    \includegraphics[scale=0.45]{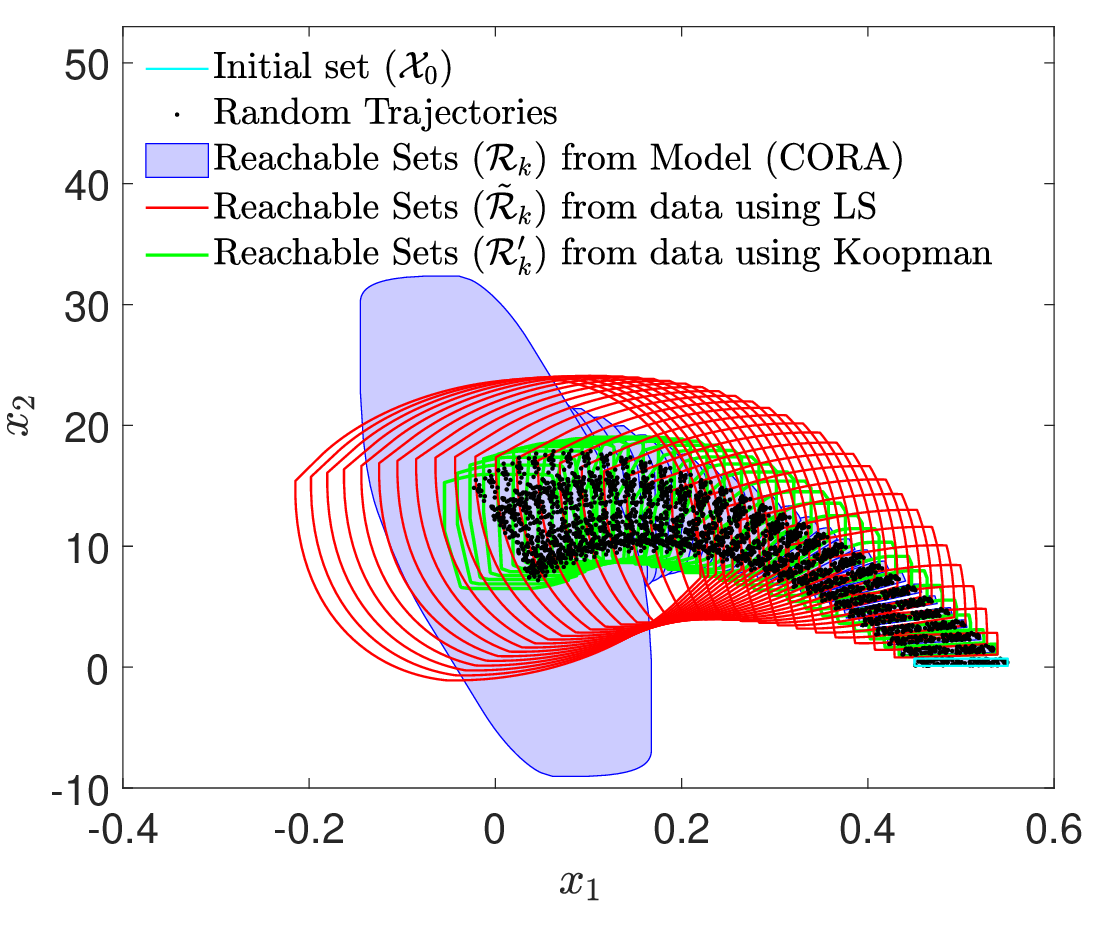}
   \caption{Comparison of computed reachable sets using CORA v2025 \cite{althoff2015introduction}, a least-squares (LS) data-driven model \cite{alanwar2021data}, and the proposed Koopman-based method, all consistent with noisy input–state data.}
    \label{fig:example0}
    \vspace{-4mm}
\end{figure}

\subsection{Example 2:Non-affine Lipschitz dynamic system}\label{ex:example2}
Consider the following discretized non-affine nonlinear system \cite{iacob2024koopman}

\begin{align*}
x_{1,k+1}
&=
x_{1,k} + T_s(\mu x_{1,k} - x_{1,k}) + T_s x_{1,k} e^{u_{1,k}}, \\\nonumber
x_{2,k+1}
&=
x_{2,k} + T_s\big(\lambda (x_{2,k} - x_{1,k}^2) - x_{2,k}\big)\\\nonumber
&\qquad + T_s\big(u_{1,k}u_{2,k} + x_{2,k} e^{u_{2,k}}\big),\nonumber
\end{align*}
where $\lambda=1$ and $\mu=-0.05$, $T_s=0.01$. For this nonlinear, non-affine system, we choose the following lifting functions

\begin{align*}
\psi(x) &= 
[1,\, x_1,\, x_2,\, x_1^2,\, x_1 x_2,\, x_2^2]^\top, \\
\nu(x,u) &= 
[\,u_1,\, u_2,\, x_1 u_1,\, x_1 u_2,\, x_2 u_1,\, x_2 u_2,\, u_1^2,\, u_1 u_2,\\
&\qquad u_2^2,\, x_1^2 u_1,\, x_2^2 u_2,\, 1\,]^\top.
\end{align*}
In Figure~\ref{fig:example2}, a comparison of reachable sets computed using a model-based approach with CORA v2025 \cite{althoff2015introduction}, a data-driven least-squares (LS) method \cite{alanwar2021data}, and the proposed method is presented. The results show that the proposed method yields less conservative over-approximations compared to the other approaches.

\begin{figure}[thp]
    \centering
    \includegraphics[scale=0.45]{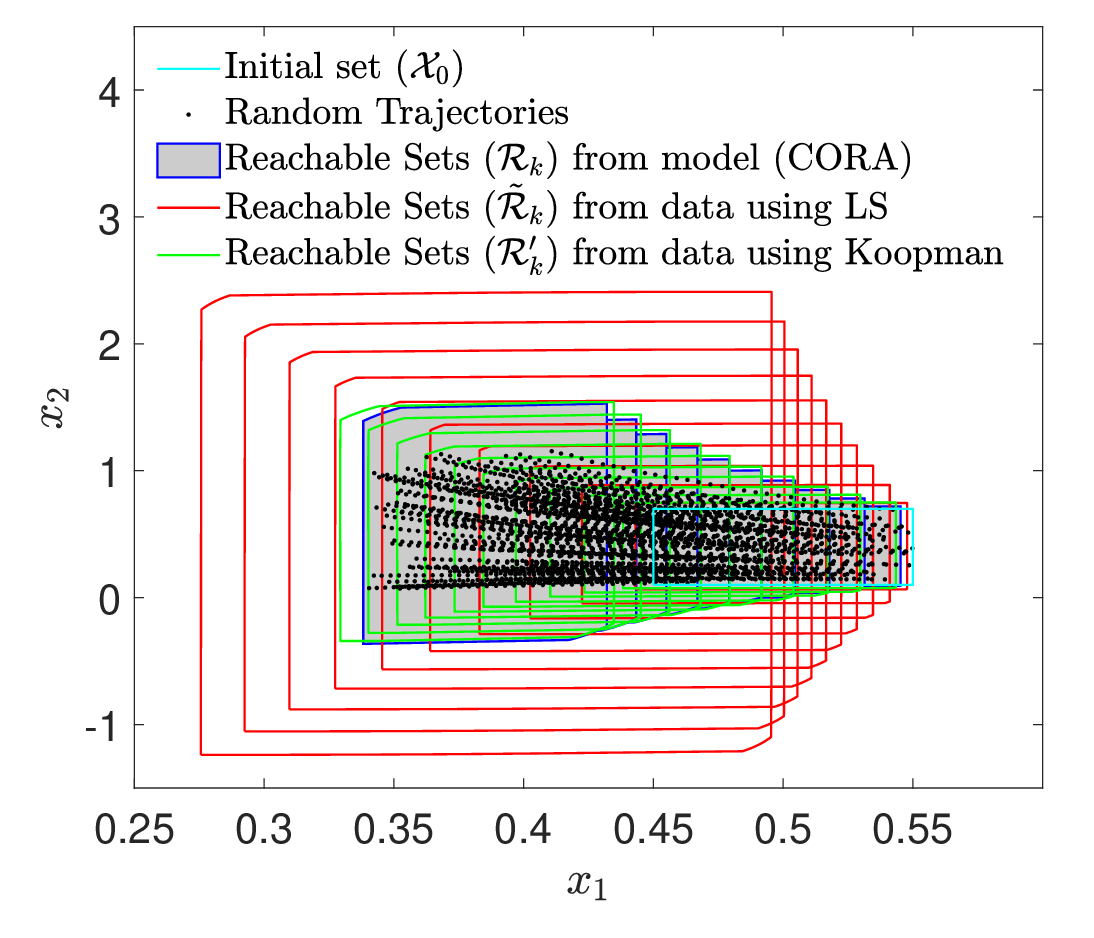}
    \caption{Comparison of computed reachable sets using CORA v2025 \cite{althoff2015introduction}, a least-squares (LS) data-driven model \cite{alanwar2021data}, and the proposed Koopman-based method, all consistent with noisy input–state data for the non-affine system.}
    \label{fig:example2}
    \vspace{-5mm}
\end{figure}

\subsection{Example 3: Autonomous Vehicle}
We utilized the JetRacer ROS AI Kit shown in Figure \ref{fig:jetracer}, an autonomous racing robot with an NVIDIA Jetson Nano Developer Kit (4GB RAM) and Raspberry Pi RP2040 dual-core microcontroller. The inputs to the vehicle are the linear ($u_1$) and angular velocity ($u_2$), and the outputs are the position ($[x_1,x_2]$) and heading ($x_3$) of the vehicle.\par
The selected lifting functions are as follows
\begin{align*}
\psi(x) &= 
[\,1,\, x_1,\, x_2,\, \sin(x_3),\, \cos(x_3),\\
&\qquad x_3,\, \sin^2(x_3),\, \cos^2(x_3)\,]^\top, \\
\nu(z,u) &= 
[\,u_1,\, u_2,\, z_5 u_2,\, z_4 u_2,\, u_2 \tan(u_1),\, u_2 u_1^2,\, z_6 u_1\,]^\top.
\end{align*}
The reachable sets computed using the method in \cite{alanwar2021data} and the proposed method are shown in Figure~\ref{fig:car1}. The proposed method is significantly less conservative over longer prediction horizons.

\begin{figure}[thp]
    \centering
    \includegraphics[scale=0.35]{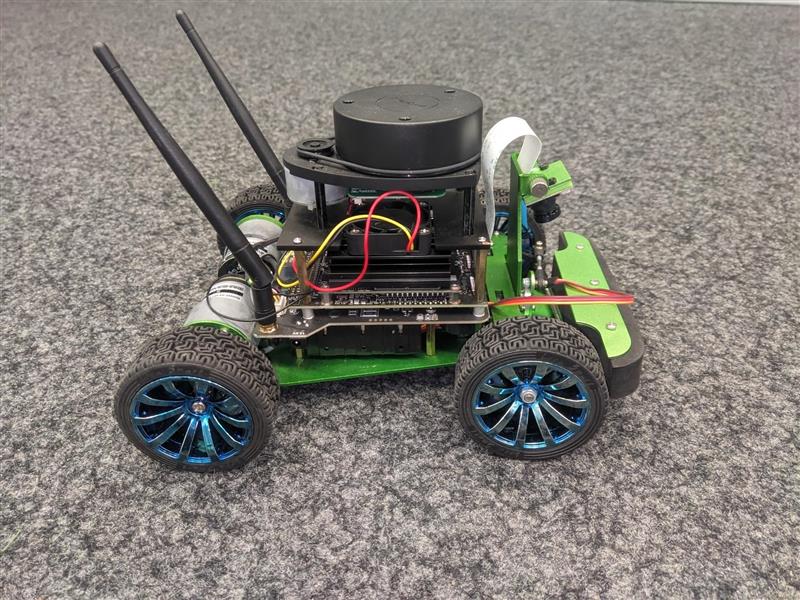}
    \caption{The JetRacer ROS AI Kit.}
    \label{fig:jetracer}
    \vspace{-4mm}
\end{figure}

\begin{figure}[thp]
    \centering
    \includegraphics[scale=0.45]{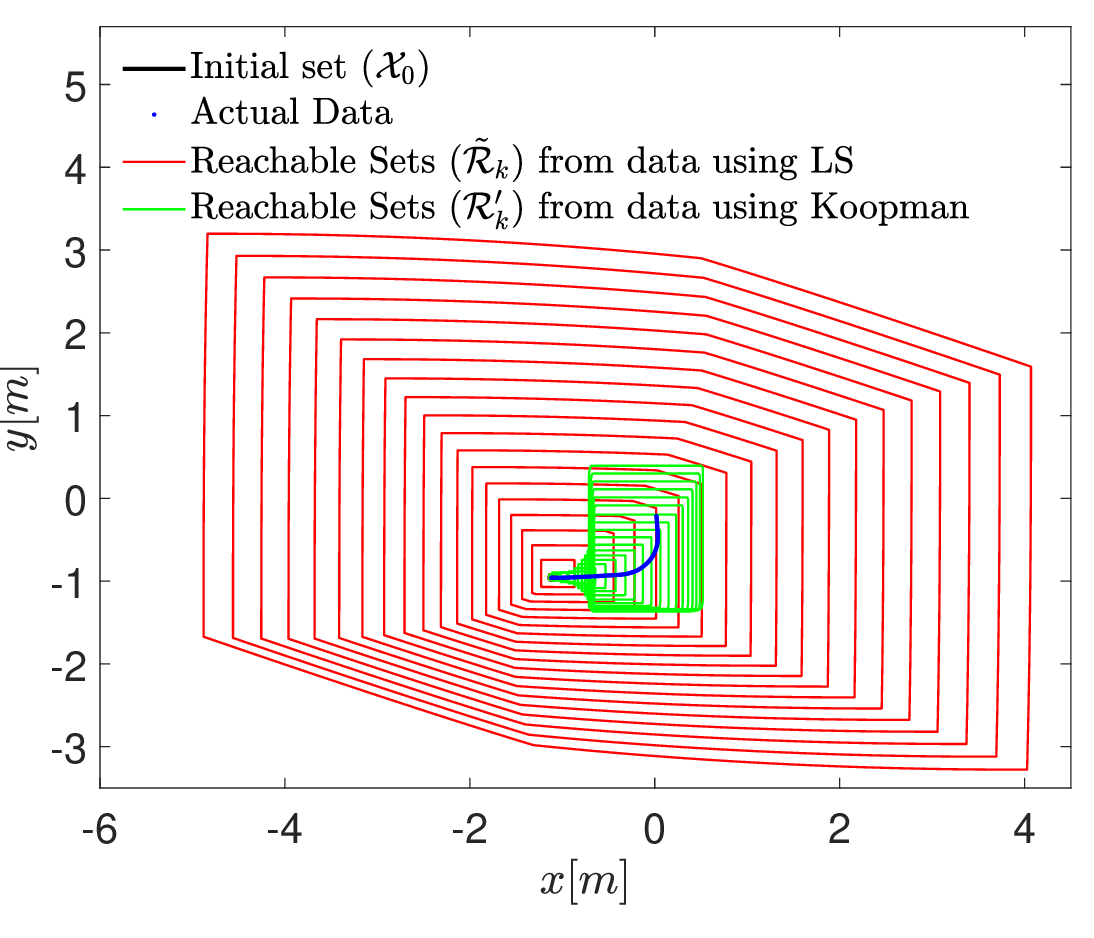}
    \caption{Comparison of reachable sets consistent with noisy input–state data using the LS method proposed in \cite{alanwar2021data} and the proposed Koopman method for real world autonomous car.}
    \label{fig:car1}
    \vspace{-5mm}
\end{figure}

\section{Conclusion}\label{sec:conclusion}
This paper presented a data-driven reachability analysis framework for nonlinear Lipschitz systems using Koopman operator embeddings under process noise. Unlike standard LTI Koopman representations, the proposed approach employs a state–input dependent lifting, enabling a more accurate linear representation of nonlinear dynamics. Reachable sets are propagated using zonotopic set representations, allowing for computationally efficient over-approximations. The results demonstrate that the proposed method yields significantly less conservative reachable sets compared to existing data-driven and model-based approaches, particularly over longer prediction horizons. This improvement stems from the ability of the state–input dependent Koopman formulation to capture nonlinear dynamics more accurately than LTI Koopman models and other linear regression-based methods. The effectiveness of the approach was validated through two numerical examples and real-world autonomous system experiments. The results consistently show that the proposed method achieves tighter and less conservative over-approximations of the reachable sets. In future works, we will extend the proposed method for time-varying nonlinear systems. 
\section{Acknowledgement}\label{sec:acknowledgement}
This work is funded by the Deutsche Forschungsgemeinschaft (DFG, German Research Foundation) – Project number: 564740977.

\bibliographystyle{IEEEtran}
{\small
\bibliography{IEEEabrv,ref}

\begin{thebibliography}{10}
\providecommand{\url}[1]{#1}
\csname url@samestyle\endcsname
\providecommand{\newblock}{\relax}
\providecommand{\bibinfo}[2]{#2}
\providecommand{\BIBentrySTDinterwordspacing}{\spaceskip=0pt\relax}
\providecommand{\BIBentryALTinterwordstretchfactor}{4}
\providecommand{\BIBentryALTinterwordspacing}{\spaceskip=\fontdimen2\font plus
\BIBentryALTinterwordstretchfactor\fontdimen3\font minus \fontdimen4\font\relax}
\providecommand{\BIBforeignlanguage}[2]{{%
\expandafter\ifx\csname l@#1\endcsname\relax
\typeout{** WARNING: IEEEtran.bst: No hyphenation pattern has been}%
\typeout{** loaded for the language `#1'. Using the pattern for}%
\typeout{** the default language instead.}%
\else
\language=\csname l@#1\endcsname
\fi
#2}}
\providecommand{\BIBdecl}{\relax}
\BIBdecl

\bibitem{ZHANG2025}
L.~Zhang, P.~Shi, S.~Wang, and X.~Su, ``A survey of safety control for service robots,'' \emph{Journal of Automation and Intelligence}, 2025.

\bibitem{Althoff2010}
M.~Althoff, ``Reachability analysis and its application to the safety assessment of autonomous cars,'' Ph.D. dissertation, Technische Universit{\"a}t M{\"u}nchen, 2010.

\bibitem{AlthoffDolan2014}
M.~Althoff and J.~M. Dolan, ``Online verification of automated road vehicles using reachability analysis,'' \emph{IEEE Transactions on Robotics}, 2014.

\bibitem{10068731}
A.~Alanwar, A.~Koch, F.~Allgöwer, and K.~H. Johansson, ``Data-driven reachability analysis from noisy data,'' \emph{IEEE Transactions on Automatic Control}, vol.~68, no.~5, pp. 3054--3069, 2023.

\bibitem{alanwar2021data}
A.~Alanwar, A.~Koch, F.~Allg{\"o}wer, and K.~H. Johansson, ``Data-driven reachability analysis using matrix zonotopes,'' in \emph{Learning for Dynamics and Control}.\hskip 1em plus 0.5em minus 0.4em\relax PMLR, 2021, pp. 163--175.

\bibitem{akhormeh2025online}
A.~N. Akhormeh, A.~Hegazy, and A.~Alanwar, ``Online data-driven reachability analysis using zonotopic recursive least squares,'' \emph{arXiv preprint arXiv:2509.17058}, 2025.

\bibitem{hafez2025safe}
A.~Hafez, A.~N. Akhormeh, A.~Hegazy, and A.~Alanwar, ``Safe llm-controlled robots with formal guarantees via reachability analysis,'' \emph{arXiv preprint arXiv:2503.03911}, 2025.

\bibitem{9833266}
M.~Selim, A.~Alanwar, S.~Kousik, G.~Gao, M.~Pavone, and K.~H. Johansson, ``Safe reinforcement learning using black-box reachability analysis,'' \emph{IEEE Robotics and Automation Letters}, vol.~7, no.~4, pp. 10\,665--10\,672, 2022.

\bibitem{girard2005reachability}
A.~Girard, ``Reachability of uncertain linear systems using zonotopes,'' in \emph{International workshop on hybrid systems: Computation and control}.\hskip 1em plus 0.5em minus 0.4em\relax Springer, 2005, pp. 291--305.

\bibitem{mamakoukas2019local}
G.~Mamakoukas, M.~Castano, X.~Tan, and T.~Murphey, ``Local koopman operators for data-driven control of robotic systems,'' in \emph{Robotics: science and systems}, 2019.

\bibitem{Str_sser_2026}
R.~Strässer, M.~Schaller, K.~Worthmann, J.~Berberich, and F.~Allgöwer, ``Safedmd: A koopman-based data-driven controller design framework for nonlinear dynamical systems,'' \emph{Automatica}, vol. 185, p. 112732, Mar. 2026.

\bibitem{proctor2018generalizing}
J.~L. Proctor, S.~L. Brunton, and J.~N. Kutz, ``Generalizing koopman theory to allow for inputs and control,'' \emph{SIAM Journal on Applied Dynamical Systems}, vol.~17, no.~1, pp. 909--930, 2018.

\bibitem{Boreale2025}
M.~Boreale and L.~Collodi, ``Linearization, model reduction and reachability in nonlinear odes,'' \emph{Formal Methods in System Design}, vol.~66, no.~3, pp. 419--454, 2025.

\bibitem{nath2026scalable}
D.~Nath, H.~Yin, and G.~Chou, ``Scalable data-driven reachability analysis and control via koopman operators with conformal coverage guarantees,'' \emph{arXiv preprint arXiv:2601.01076}, 2026.

\bibitem{KETEMA2025100339}
T.~Ketema, S.~L. Tilahun, S.~D. Zawka, and A.~Geletu, ``Deep koopman-based reachability analysis for data-driven predictive control of unknown nonlinear systems,'' \emph{IFAC Journal of Systems and Control}, vol.~34, p. 100339, 2025.

\bibitem{Bak2025}
S.~Bak, S.~Bogomolov, B.~Hencey, N.~Kochdumper, E.~Lew, and K.~Potomkin, ``Reachability of koopman linearized systems using explicit kernel approximation and polynomial zonotope refinement,'' \emph{Formal Methods in System Design}, vol.~66, no.~2, pp. 307--333, Aug 2025.

\bibitem{conf:zono1998}
W.~K{\"u}hn, ``Rigorously computed orbits of dynamical systems without the wrapping effect,'' \emph{Computing}, vol.~61, no.~1, pp. 47--67, 1998.

\bibitem{dahdah2022system}
S.~Dahdah and J.~R. Forbes, ``System norm regularization methods for koopman operator approximation,'' \emph{Proceedings of the Royal Society A}, vol. 478, no. 2265, p. 20220162, 2022.

\bibitem{iacob2024koopman}
L.~C. Iacob, R.~T{\'o}th, and M.~Schoukens, ``Koopman form of nonlinear systems with inputs,'' \emph{Automatica}, vol. 162, p. 111525, 2024.

\bibitem{kutz2016dynamic}
J.~N. Kutz, S.~L. Brunton, B.~W. Brunton, and J.~L. Proctor, \emph{Dynamic Mode Decomposition: Data-Driven Modeling of Complex Systems}.\hskip 1em plus 0.5em minus 0.4em\relax Philadelphia, PA: Society for Industrial and Applied Mathematics (SIAM), 2016.

\bibitem{Williamsetal2015}
M.~O. Williams, I.~G. Kevrekidis, and C.~W. Rowley, ``A data-driven approximation of the koopman operator: Extending dynamic mode decomposition,'' \emph{Journal of Nonlinear Science}, vol.~25, no.~6, pp. 1307--1346, 2015.

\bibitem{conf:taylor}
M.~Berz and G.~Hoffst{\"a}tter, ``Computation and application of {Taylor} polynomials with interval remainder bounds,'' \emph{Reliable Computing}, vol.~4, no.~1, pp. 83--97, 1998.

\bibitem{althoff2015introduction}
M.~Althoff, ``An introduction to cora 2015,'' in \emph{Proc. of the workshop on applied verification for continuous and hybrid systems}, 2015, pp. 120--151.

\bibitem{conf:nonlinearexample}
J.~M. Bravo, T.~Alamo, and E.~F. Camacho, ``Robust {MPC} of constrained discrete-time nonlinear systems based on approximated reachable sets,'' \emph{Automatica}, vol.~42, no.~10, pp. 1745--1751, 2006.

\end{thebibliography}
}

\end{document}